\documentclass[a4paper,11pt]{article}
\usepackage{vmargin}
\setpapersize{A4}
\setmarginsrb{15mm}{15mm}{15mm}{15mm}{0mm}{05mm}{0mm}{10mm}

\usepackage{amssymb,latexsym}
\usepackage{amsmath}
\usepackage{amsfonts}
\usepackage{graphicx}
\usepackage{subfigure}
\usepackage{latexsym}
\usepackage{amsthm}
\usepackage{amssymb}
\usepackage{color}

\theoremstyle{plain}% default

\newtheorem{proposition}{Proposition}

\renewenvironment{proof}{\noindent{\bfseries Proof.}}{\qed}

\theoremstyle{definition}
\newcommand{\inte}{{\rm int}}

\newcommand{\ben}{\begin{enumerate}}
\newcommand{\een}{\end{enumerate}}

\newcommand{\be}{\begin{equation}}
\newcommand{\ee}{\end{equation}}
\newcommand{\bery}{\begin{eqnarray}}
\newcommand{\eery}{\end{eqnarray}}

\newcommand{\R}{{\mathbb R}}

\begin{document}

\title{\textbf{A Multicriteria Macroeconomic Model with\\ Intertemporal Equity and Spatial Spillovers}}
\author{Herb Kunze\footnote{Department of Mathematics and Statistics, University of Guelph, Guelph, Canada. Email: hkunze@uoguelph.ca} \and{Davide La Torre\footnote{SKEMA Business School - Universit$\acute{e}$ C$\hat{o}$te d'Azur, Sophia Antipolis Campus, France. Email: davide.latorre@skema.edu}}
\and{Simone Marsiglio\footnote{Department of Economics and Management, University of Pisa, Pisa, Italy. Email: simone.marsiglio@unipi.it}}}

\maketitle

\begin{abstract}
We analyze a macroeconomic model with intergenerational equity considerations and spatial spillovers, which gives rise to a multicriteria optimization problem. Intergenerational equity requires to add in the definition of social welfare a long run sustainability criterion to the traditional discounted utilitarian criterion. The spatial structure allows for the possibility of heterogeneiity and spatial diffusion implies that all locations within the spatial domain are interconnected via spatial spillovers. We rely on different techniques (scalarization, $\epsilon$-constraint method and goal programming) to analyze such a spatial multicriteria problem, relying on numerical approaches to illustrate the nature of the trade-off between the discounted utilitarian and the sustainability criteria. \medskip

\noindent \textbf{Keywords:} Multicriteria Optimization; Intergenerational Equity; Spatial Diffusion
\end{abstract}

\section{Introduction}

After decades of debates a wide consensus on the effects of anthropogenetic activities on the environment has finally emerged, and even policymakers seem finally convinced that it is now time to act in order to ensure the long term sustainability of economic activities. Sustainability is a complicated notion to define but in its most widely accepted terms it requires some respect of natural resources and some efforts to ensure intergenerational equity (\cite{wced}). While the former aspect can be easily accounted for in macroeconomic analysis by adding an additional constraint to the standard optimization problem representing the society's planning device, the latter is more problematic since the criterion generally used in the definition of such a problem, the discounted utilitarianism, attaches less and less weight to future generations and so it cannot accommodate for intergenerational issues (\cite{Colapinto}). Different approaches to overcome such an issue have been proposed, but probably the most commonly used consists of extending the objective function to include an additional term representing somehow long term sustainablity considerations (\cite{Chinchilnisky}). The introduction of such an additional term in the optimization problem makes it a multiple objective problem in which the society needs to balance two conflicting goals, represented by short term and long term objectives, respectively; such an interpretation of a macroeconomic problem as a multicriteria problem allows to rely upon the operations research methods to analyze macroeconomic issues, bridging somehow these two important but distinct fields of research (\cite{MarsiglioLatorre}; \cite{MarsiglioPrivileggi}).

The goal of this paper is therefore to analyze from an operations research perspective a traditional macroeconomic model extended along two different directions. The first consists of allowing for long term sustainability considerations, in the form of intertemporal equity issues, which introduces a second criterion in the otherwise standard unicriterion macroeconomic optimization problem; several studies have analyzed in different contexts problems of this kind by focusing on pollution control and debt reduction settings (\cite{LaTorre2017}; \cite{LatorreMarsiglio}). The second consists of allowing for a spatial dimension in which different locations within the entire economy interact one another through the trade channel; several studies analyze problems of this kind, which generally are all unicriterion, by focusing on capital accumulation and environmental problems (\cite{LaTorre2015}; \cite{LaTorre2019}). Specifically, we consider a simple setting with intergenerational equity (as in \cite{Colapinto}) in which  capital evolves both over time and across space (as in \cite{Boucekkine2009}), and to the best of our knowledge our paper represents the first attempt to analyze a spatial macroeconomic problem from an operational research point of view.

Our paper relates thus to two very distinct literatures addressing the intertemporal equity issues associated with sustainabiility and the existence of spatial spillovers in macroeconomic geography contexts, respectively. Several papers discuss the problems embedded in the use of discounted utilitarianism as a social welfare criterion by proposing alternative criteria and discussing their limits in terms of applicability and in terms of existence of optimal solution paths (\cite{Ramsey}; \cite{Chinchilniskyetall}; \cite{Pezzey}; \cite{Lekama}; \cite{Arrow}; \cite{Heal}). Several other papers instead more recently discuss how the presence of a spatial dimension allows to characterize the possibility of spatial heterogeneity and spatial spillovers, along with their implications for macroeconomic outcomes (\cite{Brito}; \cite{Boucekkine2009}; \cite{Camacho2004}; \cite{Camacho2008}; \cite{LaTorre2015}). We bridge these two branches of the economics literature by developing a spatial version of the simplest intertemporal equity problem to show how an operations research approach can be used to inform our analysis.

%Economic growth models have been recently extended to a spatial dimension in order to characterize how different locations within the whole economy interact one another through the trade channel (\cite{Brito}; \cite{Boucekkine2009}; \cite{Camacho2004}; \cite{Camacho2008}; \cite{LaTorre2015}). Most of the papers focus on a Solow-type (\cite{Solow}) purely dynamic setting in which agents' behavior is exogenously given, while more limited in number are those analyzing a Ramsey-type (\cite{Ramsey}) setting in which agents optimally determine their actions. In particular, the presence of agents' optimization precludes to derive analytical results unless in specific model's formulation in which either the objective function or the dynamic constraint is linear (\cite{Boucekkine2013b}), and even the numerical analysis is not simpler to manage since the problem turns out to be ill-posed unless specific restrictive assumptions are made (\cite{Boucekkine2013a}). Our contribution to this literature is twofold: we present an extended spatial economic growth model to describe the production process within an industrial cluster in which the output produced in a given location strictly depends on the output produced in other locations as well; we present a numerical algorithm which is general enough to determine the spatio-temporal evolution of the the main variables without imposing specific restrictions on the shape of the objective function and the dynamic constraint.

The paper proceeds as follows. Section \ref{basic} recalls the basic definitions and properties of optimization theory with multiple objectives, discussing the main methods that we shall adopt in our analysis, that are scalarization, $\epsilon$-constraint method and goal programming. %Section \ref{sec:solow} discusses the model in a purely dynamic setting characterizing both analytically and numerically some interesting results.
Section \ref{modelsec} presents our multicriteria problem showing how it can be reformulated from the points of view of scalarization, $\epsilon$-constraint method and goal programming. Section \ref{numericsec} determines the solution of the problem relying on scalarization illustrating the nature of the trade-off between the discounted utilitarian and the sustainability criteria and deriving the Pareto-frontier in a specific model's parametrization. Section \ref{numsec} presents some further numerical experiments by analyzing the problem via the $\epsilon$-constraint method and goal programming, under different model's parametrizations. Section \ref{sec:conc} as usual concludes.

\section{Preliminaries on Multiple Objective Optimization} \label{basic}

This section recalls some basic facts in  Multiple Objective Optimization (MOP). In a very abstract setting, a finite-dimensional MOP problem (see \cite{Sawaragi}) takes the form
\begin{eqnarray}
\label{AP}
\max_{x\in X} \ J(x)
\end{eqnarray}
where $(X,\|\cdot\|)$ is a Banach space and $J:X\to\R^p$ is a vector-valued functional. As usual we suppose that $\R^p$ is ordered by the classical Pareto cone $\R^p_+$. A point $x\in X$ is said to be {\em Pareto optimal} or {\em efficient} if $J(x)$ is one of the maximal elements of the set of achievable values $J(X)$. Thus a point $x$ is Pareto optimal if it is feasible and, for any possible $x'\in X$, $J(x)\le_{\R^p_+} J(x')$ implies $x=x'$. In a more synthetic way, a point $x\in X$ is said to be Pareto optimal if $(J(x) + \R^p_+) \cap J(X) = \{J(x)\}$. Because of its dimensionality and the existing of conflicting criteria, a MOP model is usually difficult to be solved and the determination of the entire or part of the Pareto frontier can be very complicated and computationally intensive. In particular this applies when the number of objectives is larger than two, leading to a higher-dimensional Pareto surface. To overcome this difficulties and reduce the model complexity, several techniques have been proposed in literature. The generation of the Pareto frontier can be accomplished through one of two predominant techniques: scalarization and vectorization methods. Among the scalarization techniques, the most frequently applied are linear scalarization, the $\epsilon$-constraint method, and goal programming. These techniques will be used in the sequel of this paper. 
The vectorization algorithms, instead, tackle the MOP model directly without transforming it into some equivalent single criterion model.

\subsection{Linear Scalarization}

The linear scalarization technique (or weighted sum) is probably the simplest and the most widely used technique to solve MOP problems
and it converts the MOP model into a family of parametric single criterion optimization models.
By using this approach, a multiple objective model can be reduced to a single criterion problem by summing up all criteria with different weights. More precisely, by linear scalarization a MOP model boils down to:
\begin{eqnarray}
\label{SP}
\max_{x\in X} \sum_{i=1}^p \beta_i J_i(x)
\end{eqnarray}
where $\beta$ is a vector taking values in the interior of $\R^p_+$, namely $\beta\in\inte(\R^p_+)$. Since the Pareto optimal solution depends on $\beta$, by modifying the weights $\beta$ different points on the Pareto optimal set can be found. Linear scalarization can also be applied to problems in which the ordering cone is different from the Pareto one. In this case, we have to rely on the elements of the dual cone to scalarize the problem.

If $J$ is a vector-valued concave functional, namely each component $J_i$ is concave, then the linear scalarized problem (\ref{SP}) is also concave. This means that one can find Pareto optimal points of a concave problem by solving a concave scalar optimization problem, and for each $\beta\in\inte(\R^p_+)$, different Pareto optimal points can be obtained. For concave problems the converse of this result is only partially true, since for each Pareto optimal point $\bar x$, there is a nonzero $\bar\beta\in\R^p_+$ such that $\bar x$ is a solution of the scalarized problem (\ref{SP}) with $\beta=\bar \beta$. This is stated in the following theorem.

\begin{proposition}
(\cite{Sawaragi})
Suppose that $D$ is convex and $J_i$ are concave for all $i=1...p$. Then for all Pareto optimal solutions $\bar x$ there exists $\beta\in\R^p_+$ such that
\begin{equation}
\bar x \in argmax_{x\in D} \left\{\sum_{i=1}^p \beta_i J_i(x) \right\}.
\end{equation}
\end{proposition}

If the linear scalarization method is used for non-concave problems, the Pareto frontier generated will be incomplete and the Decision Maker (DM) will have a non-complete set of possible solutions. In this case, other scalarization methods can be found in literature and one which is worth to be mentioned is the \textit{Chebyshev scalarization model}.

\subsubsection{$\epsilon$-constraint method}

%The second model that is proposed to solve the vector-valued problem is the \textit{$\epsilon$-constraint method}.
The second model that is proposed to solve the vector-valued problem is the \textit{$\epsilon$-constraint method}.
In this method, one of the objective functions is selected to be optimized while the others are converted into additional constraints.
The method is an hybrid methodology, in fact for the $\{J_i\}_{i\neq k}$, least acceptable levels, $\epsilon_i$ have to be set while the remaining objective function $J_k$ is optimized. Then the decision maker plays a crucial role in this setting, by choosing which objective function to optimize and the least acceptable levels for the objective functions to add as constraints. Under this method, the original vector-valued problem can be now written as:

\begin{equation}
\max J_k(x)
\end{equation}
subject to:
\begin{equation}
\label{epsilon_system}
\left\{
\begin{array}{ll}
J_i(x) \geq \epsilon_i \ \ \ \ i\neq k\\
x\in X
\end{array}
\right.
\end{equation}

This method has the advantage of being theoretically able to identify Pareto optimal points also of non-convex problems. However, it also has two potential drawbacks: the identified optimal point is only granted to be weakly Pareto optimal, and the problem might become unfeasible due to the additional constraints.

\subsection{Goal Programming}

%Another method that can be used to solve vector-valued problems that is worth to be mentioned is the \textit{Goal Programming} (or GP approach).
The \textit{Goal Programming} (or GP approach) is another widely used method to deal with vector-valued problems (see \cite{Charnesetal} \cite{Charnes}).
With respect to other scalarization techniques, the idea behind this model is the determination of the aspiration levels of an objective function. A GP model does not try to find an optimal solution but an acceptable one, as it tries to achieve the goals set by the DM rather than maximizing or minimizing the objective functions. However, an optimization procedure is involved anyway. Within this formulation, one tries to minimize any possible deviation from the objective goals, either positive or negative. In fact, the GP model is a distance-function model in which the obtained optimal solution represents the best compromise between different objectives. Since the introduction of this methodology for MOP problems, many variants have been presented in literature. Among them, the most popular one is the \textit{Weighted Goal Programming} (\cite{Aouni}) which reads as follows: Given a set of ideal goals \textit{$g_i$}, with \textit{$i=1, \dots,p$}, chosen by the DM, solve the following program:
\begin{equation}
\min \sum_{i=1}^p \theta_i^+ \delta_i^+ + \theta_i^- \delta_i^- \nonumber
\end{equation}
Subject to:
\begin{equation}
\label{GPsystem}
\left\{
\begin{array}{l}
J_i(x)+\delta_i^- -\delta_i^+ = g_i   \ \ \ \  i=1,\dots, p\\
\delta_i^-,\delta_i^+ \geq 0\ \ \ \ \forall i=1,\dots,p\\
x\in X
\end{array}
\right.
\end{equation}
where $\delta_i^+$, $\delta_i^-$ are the positive and negative deviations (slack variables), respectively, and $\theta_i^+$,  $\theta_i^-$ are the corresponding weights. Due to its simplicity, the GP model and its variants have been widely applied to different areas such as accounting, marketing, human resources, production, and so on. A negative counterpart of the GP model, that it is important to consider, is the ability of GP to produce solutions that are not Pareto efficient. To overcome this difficulty, in order to produce Pareto optimal solutions the GP model is implemented within a two-steps algorithm: At first, the GP solution is tested for Pareto efficiency and, if it is not efficient, a restoration or projection method is proposed to restore efficiency at the second step.

\section{The Model} 
\label{modelsec}

The simplest macroeconomic setting to account for intergenerational equity (\cite{MarsiglioLatorre}; \cite{MarsiglioPrivileggi}) and spatial links and spillovers (\cite{Brito}; \cite{Boucekkine2009}) consists of an optimization problem in which the social planner, by considering the evolution of the capital stock $K(x,t)$, tries to determine the optimal level of consumption $C(x,t)$ over time $t\in[0,T]$ and across space $x\in\Omega$ to maximize the vector-valued social welfare which is composed of two terms. The first term represents the discounted utility stream of the representative individuals located in different venues within the entire spatial economy $J_1=\int_0^T \int_\Omega U(C(x,t))e^{-\rho t} dx dt$, where the utility function depends on consumption and the discount factor is $\rho>0$. The second term represents the payoff of the representative individuals located in different venues within the entire spatial economy at the end of the planning horizon, $J_2=\int_{\Omega} K(x,T)$ where the payoff function depends on the final capital level. The society's optimization problem can therefore be stated as a bi-criteria problem as follows:
\begin{equation}
\label{W}
\max_{C(x,t)} W=\left(\int_0^T \int_\Omega U(C(x,t))e^{-\rho t} dx dt, \int_{\Omega} K(x,T) dx\right)
\end{equation}
Subject to
\begin{equation}
\label{eq_system1}
\small{\left\{
\begin{array}{ll}
{\partial K(x,t)\over \partial t} = \nabla (d(x) \nabla K(x,t)) + f(K(x,t), C(x,t)), &  (x,t)\in \Omega\times (0,T)\\
d(x) {\partial K\over \partial n}(x) = 0, &  x\in\partial \Omega\\
K(x,0) = K_0(x). & x\in \Omega
\end{array}
\right.}
\end{equation}
where the spatio-temporal evolution of capital in (\ref{eq_system1}) is driven by the functions $f(K(x,t), C(x,t))$ with $f_K>0$ and $f_C<0$, and $\nabla (d(x) \nabla K(x,t))$ which quantifies the process of spatial diffusion between different locations within the spatial domain. Note that the function $f(K(x,t), C(x,t))$ determines the nature of the trade-off between the two terms in (\ref{W}). The first term represents the traditional criterion considered in macroeconomic analysis, characterizing the short termism of policymaking which leads the society to care about the wellbeing of the current generation. The second term represents instead a sustainability criterion introduced to account for the wellbeing of future generations by allowing the society to determine today how many resources not to exploit in order to preserve them for the future. Clearly, a higher consumption today increases the first criterion and through its subtraction of resources to capital accumulation decreases the second one. Therefore, the society needs to balance these two conflicting criteria determining the optimal dynamic path of consumption.

The above maximization has to be understood in a Pareto sense and with respect to the Pareto order in $\R^2$: Given $a,b\in\R^2$, $a\leq b$ if and only if $a_i\le b_i$, $i=1,2$. So, in other words, a feasible pair $(\bar K(x,t),\bar C(x,t))$ that solves (\ref{eq_system1}) is optimal if there is no other feasible pair $(K(x,t),C(x,t))$ such that$(J_1(K(x,t),C(x,t)),J_2(K(x,t),C(x,t)))$ dominates $(J_1(\bar K(x,t),\bar C(x,t)),J_2(\bar K(x,t),\bar C(x,t)))$. This can be rewritten by stating that there exists no feasible $(K(x,t),C(x,t))$ such that
$$
(J_1(\bar K,\bar C),J_2(\bar K,\bar C)) \in (J_1(K,C),J_2(K,C))\mathcal{} + \inte(\R^2).
$$
Using the three approaches presented in the previous section to reduce a multiple objective problem to a single objective model, we can define the following three different single-criterion formulations.

\subsection{Model I}

In this first fomulation the two criteria are combined together through scalarization weights. The DM decides the value of a trade-off parameter $\Theta\in (0,1)$, that expresses the relative importance of first criterion with respect to the second one. The parameter $\Theta$ is then expressing the level of priority and preference that each criterion has for the DM. The scalarized problem reads as
$$
\max J_\Phi(C,K) := (1-\Phi) J_1(C,K) + \Phi J_2(C,K)
$$
Subject to
$$
\small{\left\{
\begin{array}{ll}
{\partial K(x,t)\over \partial t} = \nabla (d(x) \nabla K(x,t)) + f(K(x,t), C(x,t)), &  (x,t)\in \Omega\times (0,T)\\
d(x) {\partial K\over \partial n}(x) = 0, &  x\in\partial \Omega\\
K(x,0) = K_0(x). & x\in \Omega
\end{array}
\right.}
$$
where
$$
J_1(C,K)=\int_0^T \int_\Omega U(C(x,t))e^{-\rho t} dx dt
$$
and
$$
J_2(C,K)= \int_{\Omega} K(x,T)dx
$$
If we replace $\Theta={\Phi\over 1-\Phi}$ then the objective function boils down to:
$$
\max J_\Theta(C,K) := J_1(C,K) + \Theta J_2(C,K)
$$
We will see in the following section dedicated to a numerical simulation that, by varying over $\Theta\in (0,+\infty)$, it is possible to reconstruct the Pareto frontier. In fact the objective function is concave in $(C,K)$ for any $\Theta\in (0,+\infty)$.

\subsection{Model II}

In this second formulation we proceed by using the $\epsilon$-constraint method. Within this approach there are essentially two different formulations that can be proposed, once again this decision being dependent on the relative importance of each criterion for the DM.
If DM considers the intertemporal utility as the main criterion and supposes that a certain level $\epsilon$ of the criterion has to be attained, the model boils down to:
\begin{equation}
\max J_\epsilon(C,K) := \int_0^T \int_\Omega U(C(x,t))e^{-\rho t} dx dt\\
\end{equation}
Subject to
\begin{equation}
\label{eq_system2}
\small{\left\{
  \begin{array}{ll}
    \int_{\Omega} K(x,T) dx \ge \epsilon, & \\
    {\partial K(x,t)\over \partial t} = \nabla (d(x) \nabla K(x,t)) + f(K(x,t), C(x,t)), &  (x,t)\in \Omega\times (0,T)\\
    d(x) {\partial K\over \partial n}(x) = 0, &  x\in\partial \Omega\\
    K(x,0) = K_0(x). & x\in \Omega
  \end{array}
\right.}
\end{equation}

By varying $\epsilon$ it is possible to obtain different points of the Pareto frontier. If, instead, the DM is more interested in the
sustainability criterion, the model can be written as
\begin{equation}
\max J_\epsilon(C,K) := \int_{\Omega} K(x,T) dx\\
\end{equation}
Subject to
\begin{equation}
\label{eq_system3}
\small{\left\{
  \begin{array}{ll}
  \int_0^T \int_\Omega U(C(x,t))e^{-\rho t} dx dt\ge \epsilon, & \\
    {\partial K(x,t)\over \partial t} = \nabla (d(x) \nabla K(x,t)) + f(K(x,t), C(x,t)), &  (x,t)\in \Omega\times (0,T)\\
    d(x) {\partial K\over \partial n}(x) = 0, &  x\in\partial \Omega\\
    K(x,0) = K_0(x). & x\in \Omega
  \end{array}
\right.}
\end{equation}

\subsection{Model III}

In this third formulation we rely on a GP approach. Let us suppose that $g_1$ and $g_2$ are the two goals associated with the two criteria $J_1$ and $J_2$, respectively. Within this technique, the DM is not interested in the maximization process tout court, but instead in the achievement levels of the two criteria 
$J_1$ and $J_2$.  The previous model, reformulated using the GP technique, reads as

\begin{equation}
\min J_{g_1,g_2} := \sum_{i=1}^2 \theta_i^+ \delta_i^+ + \theta_i^- \delta_i^- \nonumber
\end{equation}
Subject to
\begin{equation}
\label{eq_system3}
\small{\left\{
  \begin{array}{ll}
  \int_{\Omega} K(x,T) dx - \delta_1^+ + \delta_1^- = g_1 \\
    \int_0^T \int_\Omega U(C(x,t))e^{-\rho t} dx dt - \delta_2^+ + \delta_2^- = g_2, & \\
    {\partial K(x,t)\over \partial t} = \nabla (d(x) \nabla K(x,t)) + f(K(x,t), C(x,t)), &  (x,t)\in \Omega\times (0,T)\\
    d(x) {\partial K\over \partial n}(x) = 0, &  x\in\partial \Omega\\
    K(x,0) = K_0(x). & x\in \Omega
  \end{array}
\right.}
\end{equation}

\section{Computational Results}\label{numericsec}

In this section we present a numerical implementation of Model I introduced above and we determine an approximation of the Pareto frontier. The next section, instead, will present other possible numerical implementations of Models II and III.
Before doing that, we specify the components of this model and, in particular, the form of the function $f(K(x,t),C(x,t))$. It is a traditional Ramsey-type optimal control problem, extendeto account for intergenerational equity issues (\cite{Chinchilnisky}; \cite{Colapinto}) along with spatial heterogeneity and spillovers (\cite{Boucekkine2009}; \cite{LaTorre2015}). The economy develops along a linear city and the social planner wishes to maximize the vector-valued social welfare by choosing the level of consumption in each location, which in turn determines the evolution of capital in each location and in the whole economy. Capital accumulation depends on the difference between net (of depreciation, where the depreciation rate is $\delta_K>0$) production, $Y(x,t)= AK(x,t)^\alpha$ and consumption, augmented for the inflows of capital from other locations; these flows are captured by a diffusion term $\nabla \left(d(x) \nabla K(x,t) \right)$, where $d(x)$ is the diffusion parameter. Given the initial condition, $K(x,0) = K_0(x)$, the problem can be summarized as follows:
\begin{eqnarray}
\max & &\left(\int_0^T \int_\Omega U(C(x,t))e^{-\rho t} dx dt, \int_{\Omega} K(x,T) dx\right) \label{W1}
\end{eqnarray}
Subject to
\begin{equation}
\label{eq_system1}
\small{\left\{
  \begin{array}{ll}
    {\partial K(x,t)\over \partial t} = \nabla \left(d(x) \nabla K(x,t) \right) + A K(x,t)^\alpha - \delta_K K(x,t) - C(x,t), & (x,t)\in \Omega\times (0,T) \\
    d(x) {\partial K\over \partial n}(x) = 0, &  x\in\partial \Omega\\
    K(x,0) = K_0(x). & x\in \Omega
  \end{array}
\right.}
\end{equation}

We first discuss how to solve numerically the scalarized problem:
$$
\max J_\Theta(C,K) := J_1(C,K) + \Theta J_2(C,K)
$$
Subject to
$$
\small{\left\{
  \begin{array}{ll}
    {\partial K(x,t)\over \partial t} = \nabla \left(d(x) \nabla K(x,t) \right) + A K(y,t)^\alpha dy - C(x,t)- \delta_K K(x,t), & (x,t)\in \Omega\times (0,T) \\
    d(x) {\partial K\over \partial n}(x) = 0, &  x\in\partial \Omega\\
    K(x,0) = K_0(x). & x\in \Omega
  \end{array}
\right.}
$$
To determine an optimal policy result, let us define the current Hamiltonian function as
$$
H(C,K,\lambda) = U(C) + \lambda \left(\nabla (d(x) \nabla K(x,t) + A K(x,t)^\alpha - \delta K(x,t) - C(x,t)\right)
$$
The following proposition provides the optimality conditions for an optimal solution of the problem above.

\begin{proposition}
Suppose that $U(C)$ is a concave function of $C$. Then a pair $(\tilde C, \tilde K)$ solves the above optimal control model if and only if it is solution to the following Hamiltonian system:
$$
{\left\{
  \begin{array}{ll}
    {\partial K(x,t)\over \partial t} = \nabla \left(d(x) \nabla K(x,t) \right) + A K(x,y)^\alpha - \delta_K K(x,t) - C(x,t), & (x,t)\in \Omega\times (0,T) \\
    {\partial \lambda(x,t)\over \partial t} = \rho \lambda - \nabla \left(d(x) \nabla \lambda(x,t) \right) - \lambda \alpha A K^{\alpha-1}(x,t) - \delta_K \lambda, & (x,t)\in \Omega\times (0,T)\\
    U'(C) = \lambda & (x,t)\in \Omega\times (0,T)\\
    d(x) {\partial K\over \partial n}(x) = 0, & x\in\partial \Omega\\
    d(x) {\partial \lambda\over \partial n}(x) = 0, & x\in\partial\Omega \\
    \lambda(x,T) = \Theta & x\in \Omega \\
    K(x,0) = K_0(x) & x\in \Omega
  \end{array}
\right.}
$$
\end{proposition}

Since analyzing explicitly the Hamiltonian system above is generally not possible (unless we introduce restrictive assumptions), we now proceed with numerical simulations to illustrate the optimal behavior of capital and consumption. If we use the dynamic constraint and plug it into the objective function we obtain:
{$$
J(K) =\int_0^T \int_\Omega U\left(- {\partial K(x,t)\over \partial t} + \nabla \left(d_k(x) \nabla K(x,t) \right) +
A K(y,t)^\alpha  - \delta_K K(x,t)\right)e^{-\rho t} dx dt + \Theta \int_{\Omega} K(x,T) dx
$$}
Subject to
$$
\left\{
  \begin{array}{ll}
    -{\partial K(x,t)\over \partial t} + \nabla \left(d(x) \nabla K(x,t) \right) + A K(x,t)^\alpha - \delta_K K(x,t)\ge 0, & (x,t)\in \Omega\times (0,T) \\
    d(x) {\partial K\over \partial n}(x) = 0, & x\in\partial \Omega\\
    K(x,0) = K_0(x). & x\in \Omega
  \end{array}
\right.
$$
The direction derivative of $J$ along any feasible $h$ is given by
\begin{eqnarray*}
J'(K;h) & = &  \lim_{\delta\to 0} \frac{J(K+\delta h)- J(K)}{\delta} \\
& = & \int_0^T \int_\Omega U'(- {\partial K(x,t)\over \partial t} + \nabla \left(d(x) \nabla K(x,t) \right)  + A K^\alpha(x,t) - \delta_K K(x,t)) \\
& * & \left(-{\partial h(x,t)\over \partial t} + \nabla \left(d(x) \nabla h(x,t) \right) + A h^\alpha(x,t) dx - \delta_K h(x,t)\right)e^{-\rho t} dx dt + \Theta \int_{\Omega} h(x,T) dx \\
\end{eqnarray*}
We propose an algorithm to determine an approximation of the optimal solution. At each step this algorithm determines the direction of growth $h$ using the above calculated directional derivative $J'(K;h)$.
\begin{itemize}
  \item Given the value of the state variable $K_n(x,t)$, solve the following problem
$$
{\left\{
  \begin{array}{ll}
     - {\partial h(x,t)\over \partial t} + \nabla \left(d(x) \nabla h(x,t) \right) + A h^\alpha(x,t) - \delta_K h(x,t) = \\
      \left[U'\left(- {\partial K_n(x,t)\over \partial t} + \nabla \left(d(x) \nabla K_n(x,t) \right) + A K_n^\alpha(x,t) - \delta_K K_n(x,t)\right)\right]^{-1}
      \left(- \Theta {\partial h\over \partial t}  +1\right)e^{\rho t}
            , & (x,t)\in \Omega\times (0,T) \\
    d(x) {\partial K\over \partial n}(x) = 0, & x\in\partial \Omega\\
    h(x,0) = 0. & x\in \Omega
  \end{array}
\right.}
$$
  \item Determine $\delta>0$ that corresponds to the maximum increment of $J$ along the direction $h$
  \item Update $K_{n+1} = K_n + \delta h$
  \item If $|J(K_{n+1})-J(K_n)|<\epsilon$ then stop otherwise go to the top.
\end{itemize}

The following result shows that $J$ is increasing along the sequence generated by the above algorithm. The implementation of the above algorithm generates a sequence of functions $K_n$ along which the objective function is increasing.

\begin{proposition}
If $\delta$ is small then $J(K_{n+1}) \ge J(K_n)$, $\forall n\ge 0$.
\end{proposition}

\begin{proof}
Computing we have:
\begin{eqnarray*}
J(K_{n+1}) - J(K_n) & = & \delta J'(K_n;h) + o(\delta) \\
& = & \delta \left(\int_0^T \int_\Omega \left(- \Theta {\partial h\over \partial t}  +1\right)dx dt + \Theta \int_{\Omega} h(x,T) dx \right) + o(\delta) \\
& = & \delta  \left(- \Theta \int_\Omega h(x,T)-h(x,0) dx + T \mu(\Omega) + \Theta \int_{\Omega} h(x,T) dx \right) + o(\delta) \\
& = & \delta \left(T \mu(\Omega)+{o(\delta)\over \delta}\right)\ge 0
\end{eqnarray*}
and this last passage relies on the boundary condition $h(x,0)=0$.
\end{proof} \medskip

We now apply the above algorithm to the following model where $U(C)=[(1+C(x,t))^{2\over 3} -1]$, $T=1$, $\Omega=[0,1]$, $d_k(x)=1-0.5 x^2$, $\alpha=1$, $\delta_K=0.01$.
$$
\max \left(\int_0^1 \int_0^1 [(1+C(x,t))^{2\over 3} -1] dx dt, \int_0^1 K(x,1) dx\right)
$$
Subject to
$$
{\left\{
  \begin{array}{ll}
    {\partial K(x,t)\over \partial t} = \nabla \left((1-0.5 x^2) \nabla K(x,t) \right) + K(x,t) - C(x,t)- 0.01 K(x,t), & (x,t)\in [0,1]\times (0,1) \\
    (1-0.5 x^2) {\partial K\over \partial n}(x) = 0, &  x=0,1\\
    K(x,0) = 1+x. & x\in [0,1]
  \end{array}
\right.}
$$

We have implemented the above algorithm in COMSOL. Firstly, we run the algorithm with $\Theta=0$. The values of the objective function after three iterations are: $J_0=0.579453721074241$, $J_1=0.6059822543917376$, $J_2=0.6287663921318654$, $J_3=0.6534865860743782$. Then we used $\Theta=0.1$. The values of the objective functions after three iterations are given by: $J_0=0.5944537210742411$, $J_1=0.6200510708312351$, $J_2= 0.6404132122774402$, and $J_3=0.6570096551286734$. Figure \ref{Fig3} shows the spatio-temporal behavior of $K$ in the $\Theta=0$ case (left panel) and $\Theta=0.1$ case (right panel). It clearly shows our simple algorithm is capable to handle an optimal control problem on partial differential equations without requiring specific restrictive assumptions on the functional forms of the utility and production functions.
\begin{figure}[h]
  \centering
  \includegraphics[width=6cm]{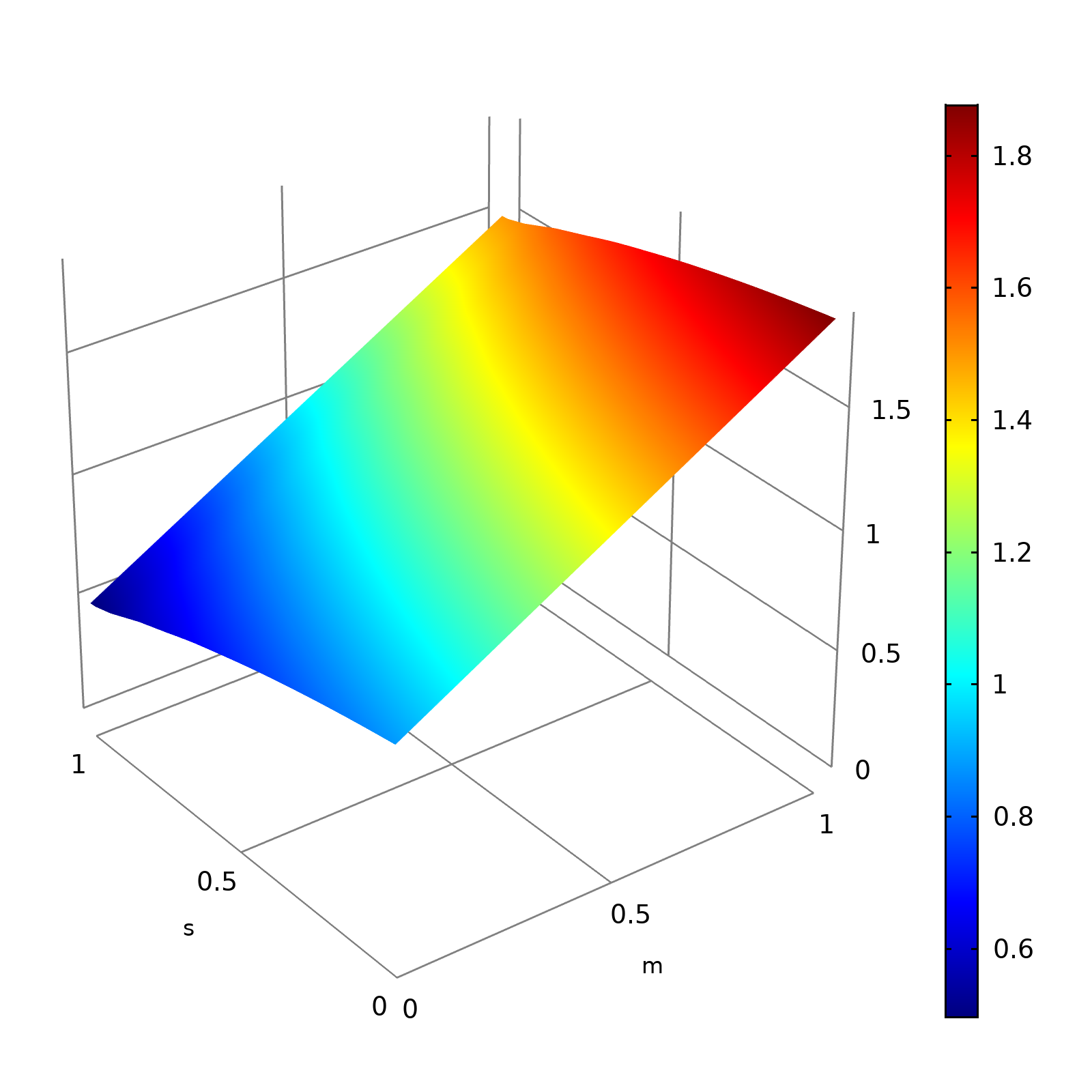}
  \includegraphics[width=6cm]{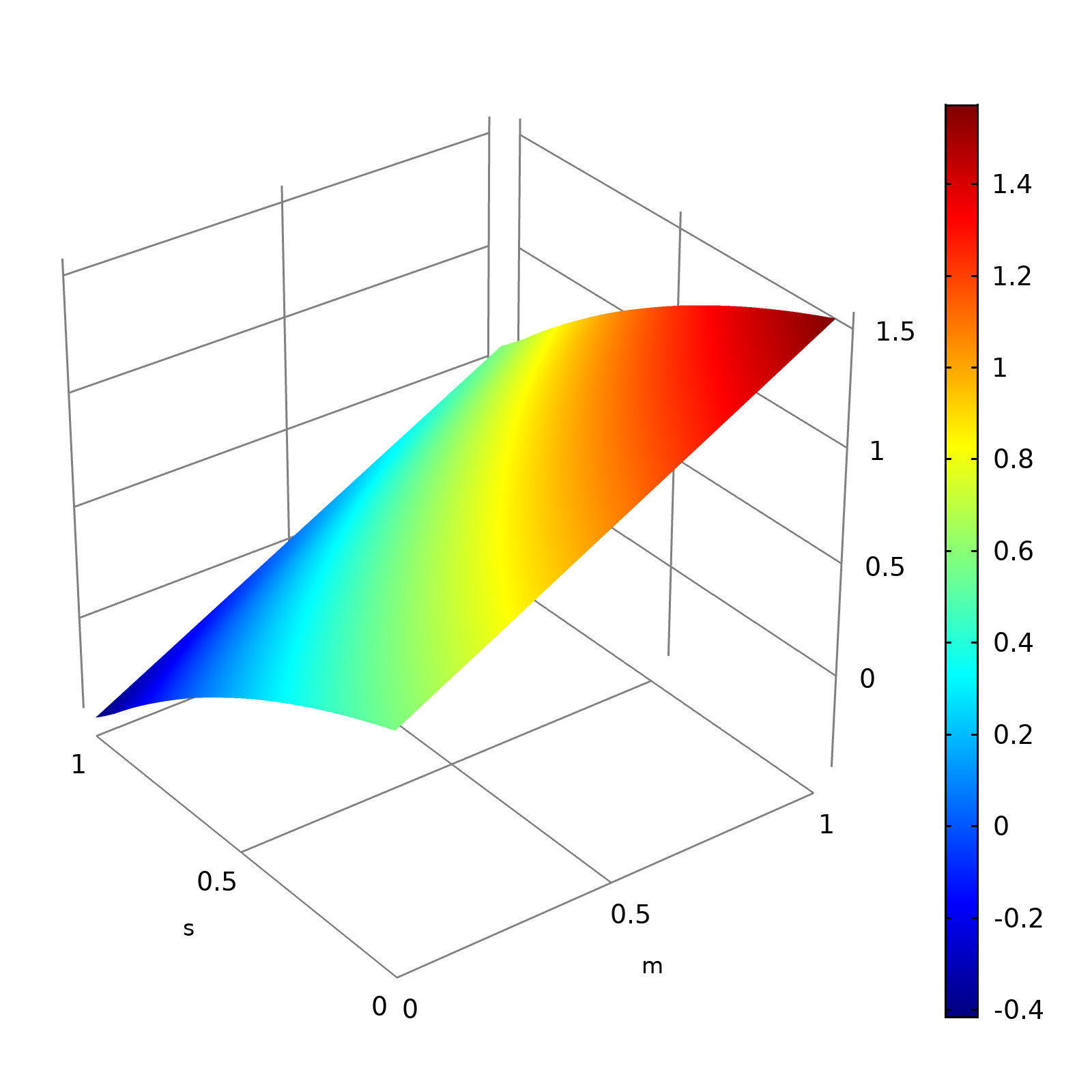}
  \caption{Spatio-temporal evolution of capital in the $\Theta=0$ case (left) and $\Theta=0.1$ case (right).}\label{Fig3}
\end{figure}

In order to construct the Pareto frontier we proceed as follows:
\begin{itemize}
  \item For each $\Theta$ we solve the “scalarized” problem $\max J_1 + \Theta J_2$,
  \item Let $K_\Theta(x,t)$ be the optimal solution, we plug it into the two separated criteria and get the pair of values $(J_1,J_2)$,
  \item We plot the pair $(J_1,J_2)$.
\end{itemize}
The Pareto frontier is illustrated in Figure \ref{Fignew}, from which we can clearly observe that it is bowed outward as expected, give the nature of the trade off between the two criteria.
\begin{figure}[h]
  \centering
  \includegraphics[width=6cm]{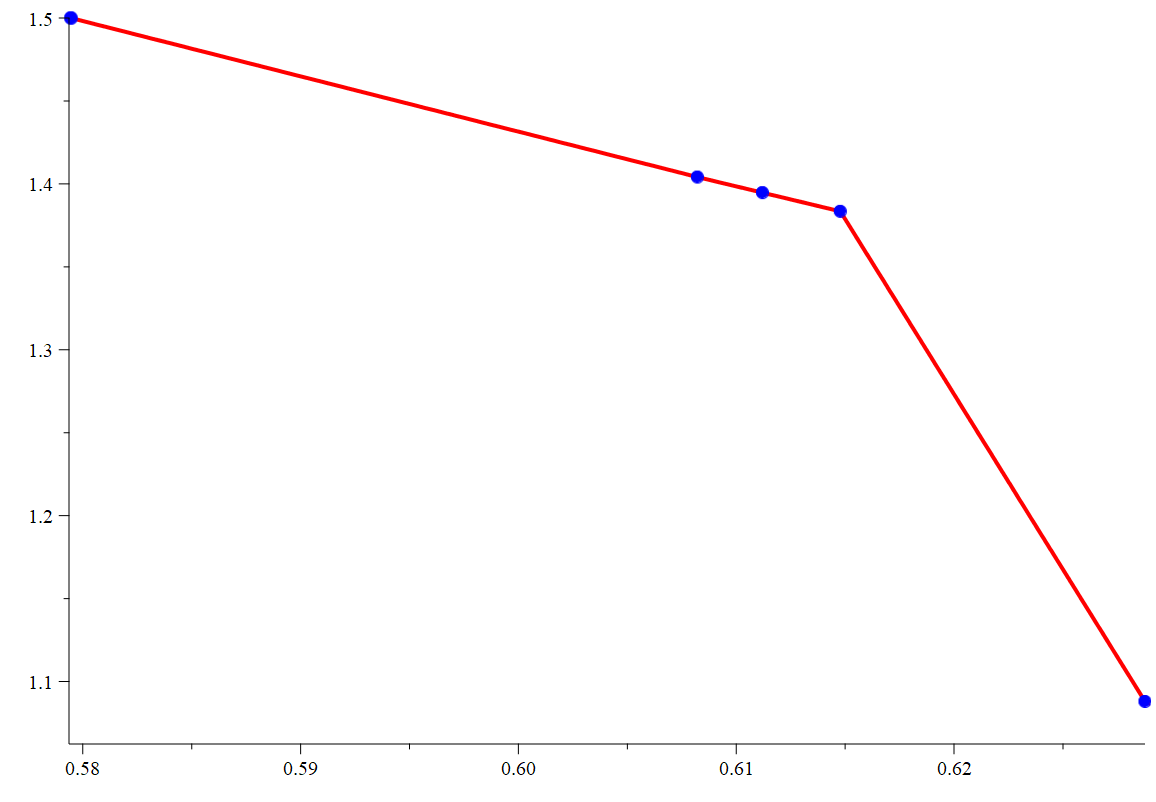}
  \caption{Piecewise linear Pareto frontier under scalarization.}\label{Fignew}
\end{figure}

\section{Further Numerical Experiments} \label{numsec}

In this section we discuss some further numerical experiments that can be implemented using different model formulations based on the $\epsilon$-constraint and the GP approaches.

\subsection{Experiment I}
\label{exp1}

We now discuss a model formulation that has been obtained by applying the $\epsilon$-constraint method. In this formulation we suppose that the DM maximizes his intertemporal utility and the level of physical capital at the final horizon $T$ is included in the set of constraints. Given a positive value of $\epsilon$, the model can be written as:

\begin{equation}
\max \int_0^T \int_\Omega U(C(x,t))e^{-\rho t} dx dt\\
\end{equation}
Subject to:
\begin{equation}
\label{eq_system2}
{\left\{
  \begin{array}{ll}
    \int_{\Omega} K(x,T) dx \ge \epsilon, & \\
    {\partial K(x,t)\over \partial t} = \nabla \left(d(x) \nabla K(x,t) \right) + A K(x,t)^\alpha - \delta_K K(x,t) - C(x,t), & (x,t)\in \Omega\times (0,T) \\
    d(x) {\partial K\over \partial n}(x) = 0, &  x\in\partial \Omega\\
    K(x,0) = K_0(x). & x\in \Omega
  \end{array}
\right.}
\end{equation}

This model can be solved using the same approach and a slightly modified version of the algorithm used in the linear scalarization case. If we use the dynamic constraint to express the level of consumption $C(x,t)$ as function of $K(x,t)$, the model boils down to:

\begin{equation}
\max \int_0^T \int_\Omega U\left({\partial K(x,t)\over \partial t} - \nabla \left(d(x) \nabla K(x,t) \right) - A K(x,t)^\alpha + \delta_K K(x,t)\right)e^{-\rho t} dx dt\\
\end{equation}
Subject to:
\begin{equation}
\label{eq_system2}
{\left\{
  \begin{array}{ll}
    \int_{\Omega} K(x,T) dx \ge \epsilon, & \\
    -{\partial K(x,t)\over \partial t} + \nabla \left(d(x) \nabla K(x,t) \right) + A K(x,t)^\alpha - \delta_K K(x,t)\ge 0, & (x,t)\in \Omega\times (0,T) \\
    d(x) {\partial K\over \partial n}(x) = 0, &  x\in\partial \Omega\\
    K(x,0) = K_0(x). & x\in \Omega
  \end{array}
\right.}
\end{equation}

As a numerical experiment, we now apply the above model when $U(C)=[(1+C(x,t))^{2\over 3} -1]$, $T=1$, $\Omega=[0,1]$, $d_k(x)=1-0.5 x^2$, $\alpha=1$, $\delta_K=0.01$. The model is formulated as:
$$
\max \int_0^1 \int_0^1 \left[\left(1+-{\partial K(x,t)\over \partial t} + \nabla \left((1-0.5 x^2) \nabla K(x,t) \right) + K(x,t) - 0.01 K(x,t)\right)^{2\over 3} -1\right] dx dt
$$
Subject to
$$
\small{\left\{
  \begin{array}{ll}
     \int_0^1 K(x,1) dx \ge 1.3 & \\
     -{\partial K(x,t)\over \partial t} + \nabla \left((1-0.5 x^2) \nabla K(x,t) \right) + K(x,t) - 0.01 K(x,t) \ge 0, & (x,t)\in [0,1]\times (0,1) \\
    (1-0.5 x^2) {\partial K\over \partial n}(x) = 0, &  x=0,1\\
    K(x,0) = 1+x. & x\in [0,1]
  \end{array}
\right.}
$$

The application of the above algorithm provides $J(1)=0.57945$,  $J(2)=0.60822$, $J(3)=0.63879$. The optimal behaviour of $K$ is shown in Figure \ref{Fignew2}.

\begin{figure}[h]
  \centering
  \includegraphics[width=6cm]{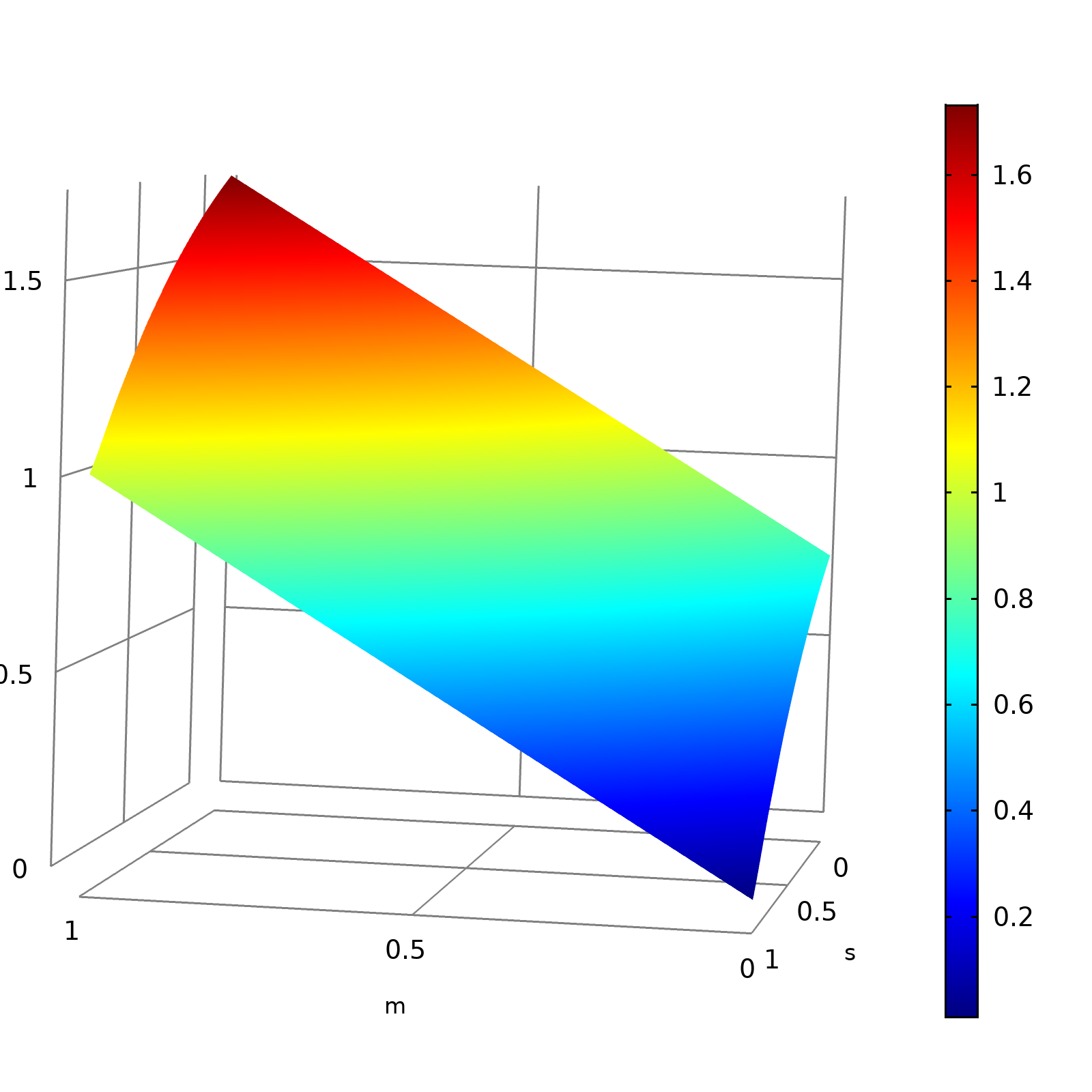}
  \caption{The evolution of $K$}\label{Fignew2}
\end{figure}

\subsection{Experiment II}
\label{exp2}

In this model formulation we still use the $\epsilon$-constraint method but we assume a linear utility function $U(C)=C$, and $\alpha=1$.
In this context the spatial model can be reduced to a one-dimensional optimal control model by introducing the average of consumption and physical capital at $t$,
$C_M(t)$ and $K_M(t)$, defined as
$$
C_M(t) = \int_\Omega C(x,t)dx
$$
and
$$
K_M(t) = \int_\Omega K(x,t) dx
$$
and the average amount of consumption per unit amount $c_M(t)$ as
$$
c_M(t) = {C_M(t)\over K_M(t)} \in [0,1]
$$
Let us suppose that a pair $(K,C)$ solves the model:

\begin{equation}
\max_{C(x,t),K(x,t)} \int_0^T \int_\Omega C(x,t)e^{-\rho t} dx dt\\
\end{equation}
Subject to:
\begin{equation}
\label{eq_system2}
{\left\{
  \begin{array}{ll}
    \int_{\Omega} K(x,T) dx \ge \epsilon, & \\
    {\partial K(x,t)\over \partial t} = \nabla \left(d(x) \nabla K(x,t) \right) + A K(x,t)^\alpha - \delta_K K(x,t) - C(x,t), & (x,t)\in \Omega\times (0,T) \\
    d(x) {\partial K\over \partial n}(x) = 0, &  x\in\partial \Omega\\
    K(x,0) = K_0(x). & x\in \Omega
  \end{array}
\right.}
\end{equation}

Then, by easy calculations and taking the integral of the contraints, the pair $(K_M,c_M)$ solves the following one:

\begin{equation}
\max_{c_M(t),K_M(t)} \int_0^T c_M(t) K_M(t) e^{-\rho t} dt\\
\end{equation}
Subject to
\begin{equation}
\label{eq_system2}
{\left\{
  \begin{array}{ll}
    \dot{K_M}(t) = A K_M(t) - \delta_K K_M(t) - c_M(t) K_M(t), & t\in (0,T) \\
    K_M(T) \ge \epsilon, & \\
    K_M(0) = \int_\Omega K_0(x) dx. &
  \end{array}
\right.}
\end{equation}

This is classical calculus of variations model with a bang-bang control and it can be solved in closed-form to determine the optimal paths of $C_M(t)$ and $K_M(t)$.

\subsection{Experiment III}
\label{exp3}

This third experiment is devoted to a different formulation that has been obtained by applying the GP approach. Here we also assume a linear utility function $U(C)=C$, and $\alpha=1$. We also suppose that all weights are equal and then normalized to $1$ and that $T$ is an integer number. Given two goals $g_1$ and $g_2$ for the two criteria $J_1$ and $J_2$, the model reads as:
\begin{equation}
\min \delta_1^+ + \delta_1^- + \delta_2^+ + \delta_2^-
\end{equation}
Subject to:
\begin{equation}
\label{eq_system3}
{\left\{
  \begin{array}{ll}
  \int_\Omega K(x,T) dx - \delta_1^+ + \delta_1^- = g_1 \\
  \int_0^T \int_\Omega C(x,t) e^{-\rho t} dx dt - \delta_2^+ + \delta_2^- = g_2, & \\
  {\partial K(x,t)\over \partial t} = \nabla \left(d(x) \nabla K(x,t) \right) + A K(x,t) - \delta_K K(x,t) - C(x,t), & (x,t)\in \Omega\times (0,T) \\
  d(x) {\partial K\over \partial n}(x) = 0, &  x\in\partial \Omega\\
  K(x,0) = K_0(x). & x\in \Omega
  \end{array}
\right.}
\end{equation}

To simplify the model, we take the integral of the PDE and reduce the analysis to the variables $K_M$ and $C_M$ introduced in the previous experiment. The model boils down to:

\begin{equation}
\min \delta_1^+ + \delta_1^- + \delta_2^+ + \delta_2^-
\end{equation}
Subject to:
\begin{equation}
\label{eq_system3}
{\left\{
  \begin{array}{ll}
  K_M(T) - \delta_1^+ + \delta_1^- = g_1 \\
  \int_0^T C_M(t) e^{-\rho t} dt - \delta_2^+ + \delta_2^- = g_2, & \\
  \dot{K_M}(t) = (A-\delta_K) K_M(t) - C_M(t), & t\in (0,T) \\
  K_M(0) = \int_\Omega K_0(x) dx. &
  \end{array}
\right.}
\end{equation}
By discretizing the time, and introducing the discrete variables $K_M(i)$, $C_M(i)$, $i=0...T$, the model can be written as
\begin{equation}
\min \delta_1^+ + \delta_1^- + \delta_2^+ + \delta_2^-
\end{equation}
Subject to:
\begin{equation}
\label{eq_system3}
{\left\{
  \begin{array}{ll}
  K_M(T) - \delta_1^+ + \delta_1^- = g_1 \\
  \sum_{i=0}^T C_M(i) e^{-\rho i} - \delta_2^+ + \delta_2^- = g_2, & \\
  K_M(i+1) - K_M(i) = (A-\delta_K) K_M(i) - C_M(i), & i=0...T-1 \\
  K_M(0) = \int_\Omega K_0(x) dx. &
  \end{array}
\right.}
\end{equation}
This is a linear optimization model that can be solved by standard optimization solvers such LINGO or MATLAB.

%\SM{Can we do some figure here for the Pareto frontier? Maybe we don't need to illustrate both the $\epsilon$-constraint and the GP methods, but one of the two %might be enough. Having the frontier for one of the two methods would allow us to compare it with that obtained under scalarization (even if the different choice of %the functional forms for $U$ might limit the scope of the comparison). }

\section{Conclusion}  \label{sec:conc}

The introduction of intergenerational equity considerations associated with sustainability issues into a traditional macroeconomic setting transforms the typical unicriterion macroeconomic problem into a bi-criteria optimization problem, which can be analyzed through the lenses of the multicriteria optimization techniques developed in the operations research literature. Recently traditional macroeconomic problems have been extended to introduce a spatial dimension, allowing to consider the extent to which spatial heterogeneity and spatial spillovers affect economic outcomes not only over time but also across space. The goal of this papers consists thus to merge these two different lines of research by analyzing a simple macroeconomic setting to account for intergenerational equity and spatial spillovers from the operations research point of view. In particular, we show that our macroeconomic problem can be reformulated as a multicriteria problem by relying on different techniques (scalarization, $\epsilon$-constraint method and goal programming), and such different formulations of the problem can be solved through numerical methods, which allow us to illustrate the nature of the trade-off between the two criteria and to derive the Pareto-frontier in some specific cases and parametrizations.

Our paper represents one of the first attempts to bridge the economics and the operations research literature, but still much needs to be done in order to develop further the possible synergies existing between these two different disciplines. In particular, apart from scalarization, $\epsilon$-constraint method and goal programming approaches, several other techniques developed in the operations research literature can be applied in similar macroeconomic contexts, especially in the context of stochastic or fuzzy multiple objective optimization. It would also be worth exploring the use of vectorization algorithms and methods that tackle the MOP model directly such as, for instance, genetic algorithms. Moreover, apart from the applications to macroeconomic questions, similar multicriteria approaches can be applied to other economic problems arising in environmental economics, game theory and cost-benefit analysis.

\end{document}